\documentclass[lettersize,journal]{IEEEtran}
\IEEEoverridecommandlockouts
\usepackage{amsmath,amssymb,amsthm,epsfig,color,empheq,graphicx,graphics,pgfplots,balance,mathtools}
\usepackage{enumerate,url,algorithm,algorithmic,wasysym,epstopdf,enumitem}
\usepackage{multirow} 
\usepackage{multicol}\usetikzlibrary{decorations.pathreplacing}
\usepackage{xcolor}
\usepackage{cuted}
\usepackage{cite}
\usepackage{amsfonts}
\usepackage{dsfont}
\usepackage{float}
\usepackage{dblfloatfix}
\usepackage{xspace}
\usepackage[font=footnotesize]{caption}
\usepackage{subcaption}
\usepackage{mathrsfs}  
\usepackage{amsmath}
\usepackage{soul}
\usepackage{multirow}
\usepackage{amsmath}
\usepackage{booktabs}
\DeclareMathAlphabet{\mathscrbf}{OMS}{mdugm}{b}{n}
\pgfplotsset{compat=1.18}

\newtheorem{theorem}{Theorem}

\begin{document}
\newcommand{\MATLAB}{\textsc{Matlab}\xspace}
\newcommand{\real}{\mathscr{R}_{\text{\footnotesize$e$}}}
\newcommand{\imag}{\mathscr{I}_{\text{\footnotesize$m$}}}

\title{Small-signal Stability of a Unified Single-unit Infinite-bus Swing-equation Model for Generators and Inverters}
  
\author{Debjyoti Chatterjee, Nathan Baeckeland, Bala Kameshwar Poolla, Gab-Su Seo, Brian Johnson, and Sairaj Dhople\\\thanks{
This work was authored by the National Laboratory of the Rockies (NLR) for the U.S. Department of Energy (DOE), operated under Contract No. DE-AC36-08GO28308. Funding provided in part by the U.S. Department of Energy Office of Energy Efficiency and Renewable Energy under the Solar Energy Technologies Office Award Number 38637 (UNIFI Consortium) and by the Laboratory Directed Research and Development (LDRD) Program at NLR. The views expressed do not necessarily represent those of the DOE or the U.S. Government. The U.S. Government retains, and the publisher, by accepting the article for publication, acknowledges that the U.S. Government retains a nonexclusive, paid-up, irrevocable, worldwide license to publish or reproduce this work for Government purposes. D. Chatterjee and B. Johnson are with the University of Texas at Austin; N. Baeckeland, B.K. Poolla, and G.-S. Seo are with the Power Systems Engineering Center of NLR; and S. Dhople is with the University of Minnesota.}}

\maketitle
\begin{abstract} We present a swing-equation model with generalized and equilibria-dependent inertia, damping, and synchronization constants for energy conversion interfaces with second-order active-power versus voltage-phasor-angle dynamics connected to an infinite bus. The model is unified in that prudent parameterization of the second-order angle-to-power transfer function aligns with reduced-order models for synchronous generators, grid-following inverters with fast frequency-response capability, and droop- and virtual synchronous generator-based grid-forming inverters. Parametric necessary and sufficient conditions to examine small-signal stability of angle equilibria are derived from the unified swing-equation model. 
\end{abstract}

\begin{IEEEkeywords}
Grid-following inverter, grid-forming inverter, small-signal stability, swing equation, synchronous generator.
\end{IEEEkeywords}

\allowdisplaybreaks

\section{Introduction}\label{sec:Intro}
\IEEEPARstart{S}{mall-signal} stability of angle equilibria for the single-machine infinite-bus (SMIB) model is a foundational tenet of power systems dynamics. While the setup is not remotely representative of complexity encountered in practice, the treatment sheds light on the interplay of power-flow equations and machine dynamics as they relate to an elementary notion of stability. Stability assessment has assumed increased importance with the changing resource mix on grids worldwide that now feature increasing numbers of grid-following~(GFL) and grid-forming~(GFM) inverter-based resources~(IBRs) alongside synchronous generators (SGs)~\cite{Lin-ImpactStability-2021}. Networks with varying complexity---including elemental single-unit models---have been in focus; however, much of the recent work on swing-equation-type models and stability of the SMIB setup in the context of IBRs is unit-specific. The main contribution of this effort is to update the ubiquitous SMIB model to one that can acknowledge the dynamics of contemporary frequency-responsive power electronics alongside SGs while preserving a pathway to analytically examine small-signal stability.

We study the quasi-steady-state dynamics of a controlled-voltage-source model for an energy-conversion interface connected to an infinite bus through an $RL$ line. It is assumed to inject a nominal amount of (commanded) active power modulated by a linear frequency-versus-power relationship. A second-order transfer function is presumed for the active power versus voltage-phasor angle, and the network interconnection is modeled via power-flow equations. A second-order swing-equation model that captures the small-signal angle dynamics emerges upon linearization about an equilibrium point. Unlike the classical swing equation for SGs, this version is generalized and sports equilibria-dependent inertia, damping, and synchronization constants. Furthermore, the model is unified in the sense that appropriate parameterization of the second-order angle-to-power transfer function aligns with reduced-order models for SGs~\cite{Ajala-2020}, GFL inverters with fast frequency response capability~\cite{Green-GFL_GFM-2023}, and droop- and virtual synchronous generator-based GFM inverters~\cite{Suul-Equivalence-2014}. 

Application of the Routh-Hurwitz criterion to the unified swing-equation model reveals necessary and sufficient conditions of unit and network parameters for examining the small-signal stability of angle equilibria. These conditions reveal that low-angle (high-angle) equilibria are nominally (un)stable; thereby aligning with long-standing results for SGs~\cite{kundur1994} and recent forays in transient stability assessment of GFL~\cite{Oriol-2025} and GFM IBRs~\cite{8026165}. By putting forth a unified modeling and small-signal stability assessment framework for all three energy-conversion interfaces, our work directly addresses a missing element in the literature on the topic. While the heterogeneous units present a seemingly uniform small-signal model, differences become evident on closer inspection of stability conditions. In particular, we find that GFL inverters with fast frequency response capability curiously offer a parametric regime in which the high-angle equilibrium can be stabilized.

To support our analytical contributions and validate the parametric necessary and sufficient conditions, we provide: i)~phase-portrait plots demonstrating dynamic behavior around equilibria for the reduced-order models, and ii)~eigenvalue plots obtained by linearizing higher-order models for SGs ($19\mathrm{th}$ order: featuring a round-rotor quadratic machine model, single-mass shaft model, automatic voltage regulator, turbine governor, and power system stabilizer), GFL inverters ($12\mathrm{th}$ order: featuring PI controllers for active and reactive power, inner current control, a synchronous reference frame phase-locked loop, and output $LCL$ filter), and GFM inverters ($14\mathrm{th}$ order: featuring a primary controller, outer voltage and inner current controllers, and output $LCL$ filter). 

The remainder of this paper is organized as follows. Section~\ref{sec:UnitModel} presents the model for an individual unit. Section~\ref{sec:Unified_SMIB} demonstrates how such a model reveals a second-order swing equation upon considering second-order unit dynamics and algebraic power-flow equations. Building on this, we outline the parametric necessary and sufficient conditions. Models for SGs, GFLs, and GFMs are presented, and stability conditions are identified. Validation covering phase portraits for the reduced-order models (examined analytically) and eigenvalues for higher-order models (discussed above) is provided in Section~\ref{sec:Validation}. Concluding remarks and a few pertinent directions for future work are given in Section~\ref{sec:ConclusionsFutureWork}.  

\section{Unit Model} \label{sec:UnitModel}
This section begins with a generalized representation of unit dynamics in the frequency domain. Nonlinear algebraic equations that describe equilibria and linearized dynamics around equilibria in generalized form are identified from the unit dynamics. 

\subsection{Generalized Frequency-domain Representation} 
The synchronous electrical radian frequency is denoted by $\omega_\mathrm{s}$. Let $E \angle\delta$ denote the line-to-neutral terminal voltage of the unit in a per-phase equivalent representation and $V_\infty \angle{0}$ denote the line-to-neutral voltage of the infinite bus. We suppose the unit is receptive to an active-power reference from a higher-level controller, $P^\star$. Furthermore, suppose $P(\delta)$ represents the nonlinear relationship between active power versus voltage-phasor angle. Let 
\begin{equation*} \label{eq:Laplace-P-delta}
    \mathcal{D}(s) = \mathscr{L}\{\delta(t)\},\,\mathcal{P}_\delta(s) = \mathscr{L}\{P(\delta(t))\}, \mathcal{P}^\star(s) = \mathscr{L}\{P^\star\}
\end{equation*}
denote pertinent Laplace transforms. The frequency at the terminals of the unit is given by
\begin{equation}
    \Omega(s) = \mathscr{L}\{\omega(t) = \tfrac{\mathrm{d}}{\mathrm{d}t}\delta(t) + \omega_\mathrm{s}\}= s\mathcal{D}(s) + \tfrac{1}{s}\omega_\mathrm{s}.
\end{equation}
Internal controls of the unit that map the power injection at the terminals to the voltage-phasor angle are captured via the angle-to-power transfer function $\mathcal{G}(s)$. The large-signal power versus angle relationship is governed by:
    \begin{equation} \label{eq:delta-p-Laplace}
        \mathcal{G}(s) \mathcal{D}(s) = \mathcal{P}^\star(s) - \mathcal{P}_\delta(s).
    \end{equation}

\subsection{Equilibrium Behavior} 
For a signal $f(t)$ with Laplace transform $\mathcal{F}(s)$, the Final Value Theorem indicates $\lim_{t\to \infty} f(t)=\lim_{s\to 0} s\, \mathcal{F}(s)$. Let $\delta(t)=\delta_\mathrm{eq}$ denote an equilibrium of the incipient dynamics and let $P(\delta_\mathrm{eq})$ correspond to the active power flow at this equilibrium. Applying the Final Value Theorem, these are
\begin{subequations} \label{eq:FVT-terms}
\begin{align}
\delta_{\rm eq}&=\lim_{t\to \infty} \delta(t)=\lim_{s\to 0} s\, \mathcal{D}(s),\\
P(\delta_{\rm eq})&=\lim_{t\to \infty} P(\delta(t))=\lim_{s\to 0} s\,\mathcal{P}_\delta(s).
\end{align}
\end{subequations}
Consider~\eqref{eq:delta-p-Laplace} in the limit:
\begin{align*}
\lim_{s\to 0} s\,\mathcal{G}(s) \mathcal{D}(s) = \lim_{s\to 0}\big(s\,\mathcal{P}^\star(s) -s\, \mathcal{P}_\delta(s) \big),
\end{align*}
which, assuming pertinent limits are well defined, can be decomposed as: 
\begin{equation*} 
\lim_{s\to 0} \mathcal{G}(s) \cdot \lim_{s\to 0} s\,\mathcal{D}(s) = \lim_{s\to 0}s\,\mathcal{P}^\star(s) -\lim_{s\to 0}s\, \mathcal{P}_\delta(s).
\end{equation*}
Applying the definitions reported in~\eqref{eq:FVT-terms} and recognizing that $P^\star$ is a constant, we get the nonlinear algebraic equation
\begin{equation} \label{eq:delta-p-steadystate2}
        \mathcal{G}(0)\, \delta_\mathrm{eq} = P^\star - P(\delta_\mathrm{eq}),
\end{equation}
from which the equilibrium $\delta_\mathrm{eq}$ can be solved for. 

\subsection{Linearized Dynamics} 
As a next step, we linearize~\eqref{eq:delta-p-Laplace} around equilibria identified by the solution of~\eqref{eq:delta-p-steadystate2}. To that end, assume a small perturbation $\delta(t) = \delta_\mathrm{eq} + \Delta \delta(t)$, which translates in the Laplace domain to $\mathcal{D}(s) = \tfrac{1}{s}\delta_\mathrm{eq} + \Delta \delta(s)$, where $\Delta \delta(s) = \mathscr{L}\{\Delta\delta(t)\}$. Substituting this in~\eqref{eq:delta-p-Laplace}, we obtain
\begin{equation} \label{eq:Linear}
    \mathcal{G}(s) (\tfrac{1}{s}\delta_\mathrm{eq} + \Delta \delta(s)) = \mathcal{P}^\star(s) - \mathscr{L}\{P(\delta_\mathrm{eq} + \Delta \delta(t))\}.
\end{equation}
Applying a first-order Taylor-series expansion for $P(\cdot)$ allows us to approximate $\mathscr{L}\{P(\delta_\mathrm{eq} + \Delta \delta(t))\}$ as follows:
\begin{align}
    \mathscr{L}\{P(\delta_\mathrm{eq} + \Delta \delta(t))\} &\approx \mathscr{L}\{P(\delta_\mathrm{eq}) + \tfrac{\partial}{\partial \delta}P(\delta)\vert_{\delta = \delta_\mathrm{eq}}\Delta \delta(t)\} \nonumber \\
    &= \tfrac{1}{s}P(\delta_\mathrm{eq})+ \tfrac{\partial}{\partial \delta}P(\delta)\vert_{\delta = \delta_\mathrm{eq}}\Delta \delta(s).
\end{align}
Substituting the above in~\eqref{eq:Linear}, recognizing that $\mathcal{G}(s)\, \delta_\mathrm{eq} = P^\star - P(\delta_\mathrm{eq})$, and rearranging terms yields the linear model: 
\begin{equation}\label{eq:delta-p-Laplace-linearized}
    \Big(\mathcal{G}(s) + \tfrac{\partial}{\partial \delta}P(\delta)\vert_{\delta = \delta_\mathrm{eq}} \Big) \Delta \delta(s) = 0.
\end{equation}

\section{Unified Swing-equation Model \& Stability} \label{sec:Unified_SMIB}
In this section, we instantiate the unit models via generalized second-order transfer functions and adopt algebraic power-flow equations to model the connection of the units with the infinite bus. This reveals the unified swing-equation model and parametric stability conditions follow therefrom by applying the Routh-Hurwitz criterion. The conditions are tailored to reduced-order models for SGs, GFLs, and GFMs.

\subsection{Integrating Unit \& Network Models}
For the units, we presume the following general form for the voltage-phasor angle to active power transfer function: 
\begin{equation}
\label{eq:angle-power-tf}
    \mathcal{G}(s) = \frac{b_2 s^2 + b_1 s }{a_2 s^2 + a_1 s + a_0}, 
\end{equation}
where $a_i, b_i$ are non-negative constants. In Section~\ref{eq:resources}, we dwell on how reduced-order models for SGs, GFL IBRs, and GFM IBRs are of the above unified form. For the network interconnection, we suppose an $RL$ line with resistance, $R$, and reactance, $X=\omega_\mathrm{s}L$. The power flow is then
\begin{equation} \label{eq:P-delta}
    P(\delta) = E^2 \tfrac{R}{R^2 + X^2} - EV_\infty \tfrac{1}{R^2 + X^2}(R \cos\delta  - X \sin\delta). 
\end{equation}
As a standing assumption to ensure feasibility, we will suppose 
\begin{equation} \label{eq:P-feasible}
    P^\star \leq \max_{\delta} P(\delta) = E^2 \tfrac{R}{R^2 + X^2} + EV_\infty \tfrac{1}{\sqrt{R^2 + X^2}}.
\end{equation}
In what follows, we will require the power-flow sensitivity:
\begin{equation} \label{eq:P-sensitivity}
    \tfrac{\partial}{\partial \delta}P(\delta) = E V_\infty\tfrac{1}{R^2 + X^2}(R\sin \delta + X \cos\delta).
\end{equation}
Figure~\ref{Fig:SingleResourceInfiniteBus} illustrates the single-unit infinite-bus model described thus far. 

\begin{figure}[t!]
    \centering
    \includegraphics[scale= 1]{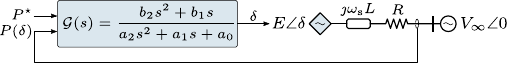}
    \caption{Single-unit infinite-bus setup includes an energy-conversion interface with a second-order angle-to-power transfer function.}
    \label{Fig:SingleResourceInfiniteBus}
\end{figure}

Substituting for $\mathcal{G}(s)$ from~\eqref{eq:angle-power-tf} and for $\tfrac{\partial}{\partial \delta}P(\delta)$ from~\eqref{eq:P-sensitivity} in~\eqref{eq:delta-p-Laplace-linearized}, simplifies it to the \textit{unified swing-equation model}: 
\begin{equation} \label{eq:unified-swing}
    M(\delta_\mathrm{eq}) s^2 \Delta \delta(s) + D(\delta_\mathrm{eq}) s \Delta \delta(s) + K(\delta_\mathrm{eq}) \Delta\delta(s) = 0,
\end{equation}
where $M(\delta_\mathrm{eq})$, $D(\delta_\mathrm{eq})$, and $K(\delta_\mathrm{eq})$ are the corresponding inertia, damping, and synchronization constants, given by
\begin{subequations}\label{eq:MDK}
\begin{align} 
    M(\delta_\mathrm{eq}) &= b_2 + a_2 \tfrac{\partial}{\partial \delta}P(\delta)\vert_{\delta = \delta_\mathrm{eq}}, \label{eq:M}\\
    D(\delta_\mathrm{eq}) &= b_1 + a_1 \tfrac{\partial}{\partial \delta}P(\delta)\vert_{\delta = \delta_\mathrm{eq}}, \label{eq:D}\\
    K(\delta_\mathrm{eq}) &= a_0 \tfrac{\partial}{\partial \delta}P(\delta)\vert_{\delta = \delta_\mathrm{eq}}.\label{eq:K}
\end{align}
\end{subequations}

\subsection{Parametric Conditions for Small-signal Stability}
With these elements in place, we next state and prove our main result on small-signal stability of equilibria referenced in the unified swing-equation model. 

\begin{theorem}
    The unit model~\eqref{eq:delta-p-Laplace} with angle-to-power transfer function given by the general form~\eqref{eq:angle-power-tf} and power flow captured by~\eqref{eq:P-delta} admits equilibria: 
    \begin{subequations} \label{eq:delta_eq12}
    \begin{align}
    \delta_{\mathrm{eq},1} &= \sin^{-1} \alpha + \tan^{-1}\Big(\tfrac{R}{X}\Big), \label{eq:delta_eq1}\\
    \delta_{\mathrm{eq},2} &= \pi- \sin^{-1} \alpha + \tan^{-1}\Big(\tfrac{R}{X}\Big),  \label{eq:delta_eq2}\\
    \quad \mathrm{where} & \,\,\,\,\alpha = \Big(\tfrac{P^\star (R^2+X^2)-E^2 R }{E V_\infty \sqrt{R^2 + X^2}}\Big). \label{eq:alpha}
    \end{align}
    \end{subequations}
    Equilibrium point $\delta_{\mathrm{eq},1}$ is small-signal stable if and only if
    \begin{align}
    \label{eq:SEPstability}
    \begin{split}
        a_0 \,(b_2 + a_2 (\tfrac{EV_\infty}{\sqrt{R^2 + X^2}}\sqrt{1-\alpha^2}) &> 0, \,\, \mathrm{and}\\
        a_0\,(b_1 + a_1 (\tfrac{EV_\infty}{\sqrt{R^2 + X^2}}\sqrt{1-\alpha^2}) &> 0.
    \end{split}
    \end{align}
    Equilibrium point $\delta_{\mathrm{eq},2}$ is small-signal stable if and only if
    \begin{align}
      \label{eq:UEPstability}
    \begin{split}
        a_0 (b_2 - a_2 (\tfrac{EV_\infty}{\sqrt{R^2 + X^2}}\sqrt{1-\alpha^2}) &< 0, \,\, \mathrm{and}\\
        a_0 (b_1 - a_1 (\tfrac{EV_\infty}{\sqrt{R^2 + X^2}}\sqrt{1-\alpha^2}) &< 0.
    \end{split}
    \end{align}
\end{theorem}

\begin{proof} Equilibria of the involved dynamics are given by the solution of~\eqref{eq:delta-p-steadystate2}. From~\eqref{eq:angle-power-tf}, we note $\mathcal{G}(0)=0$; consequently,~\eqref{eq:delta-p-steadystate2} simplifies to $P^\star = P(\delta_\mathrm{eq})$. Substituting for $P(\delta_\mathrm{eq})$ from~\eqref{eq:P-delta} and with the aid of elementary trigonometric identities, it emerges that the desired equilibria are $\delta_{\mathrm{eq},1}$ and $\delta_{\mathrm{eq},2}$ as reported in~\eqref{eq:delta_eq1}--\eqref{eq:delta_eq2}. (Note that the assumption on feasible power flow in~\eqref{eq:P-feasible} translates to $\alpha$ as defined in~\eqref{eq:alpha} satisfying $|\alpha|<1$.) Next, we consider small-signal stability of the equilibria, for which we will need an expression for the sensitivity $\tfrac{\partial}{\partial \delta}P(\delta)\vert_{\delta = \delta_\mathrm{eq}}$. 

Substituting for $\delta_{\mathrm{eq},1}$ and $\delta_{\mathrm{eq},2}$ in~\eqref{eq:P-sensitivity} followed by further trigonometric manipulations reveals
\begin{equation} \label{eq:SEP_UEP_partial}
    \tfrac{\partial}{\partial \delta}P(\delta) = \left\{
\begin{aligned}
+\psi > 0, \, & \delta = \delta_{\mathrm{eq},1}\\
-\psi < 0, \, & \delta = \delta_{\mathrm{eq},2}
\end{aligned}
\right.,\,\, \psi = \tfrac{EV_\infty \sqrt{1-\alpha^2}}{\sqrt{R^2 + X^2}}.
\end{equation} 
A necessary and sufficient condition for the roots of~\eqref{eq:unified-swing} to have negative real parts can be obtained from the second-order Routh-Hurwitz condition; in this case, simply resolving to 
\begin{equation} \label{eq:stability}
    \mathrm{sgn}(M(\delta_\mathrm{eq})) = \mathrm{sgn}(D(\delta_\mathrm{eq})) = \mathrm{sgn}(K(\delta_\mathrm{eq})).
\end{equation}
The above conditions can be expressed alternatively as
     \begin{equation} \label{eq:StabilityCompact}
    M(\delta_\mathrm{eq})\cdot K(\delta_\mathrm{eq})> 0 \,\,\mathrm{and}\,\,D(\delta_\mathrm{eq})\cdot K(\delta_\mathrm{eq})> 0.
    \end{equation}
Substituting the expressions~\eqref{eq:SEP_UEP_partial} in~\eqref{eq:MDK} and following up with~\eqref{eq:StabilityCompact} yields the stability conditions in~\eqref{eq:SEPstability} and \eqref{eq:UEPstability}.
\end{proof}

\begin{table}[t]
\centering
\caption{Small-signal stability conditions for the two equilibria.}
\label{tab:stability}
\renewcommand{\arraystretch}{1.5}
\begin{tabular}{c||c|c}
\toprule
\midrule
\textbf{unit} &
\textbf{\boldmath$\delta_{\mathrm{eq},1}$ stable iff} &
\textbf{\boldmath$\delta_{\mathrm{eq},2}$ stable iff} \\
\midrule

SG &
$\begin{array}{c}
\tfrac{2H}{\omega_\mathrm{s}}>0,\;\;
\tfrac{D}{\omega_\mathrm{s}}>0
\end{array}$ &
$\begin{array}{c}
\tfrac{2H}{\omega_\mathrm{s}}<0,\;\;
\tfrac{D}{\omega_\mathrm{s}}<0
\end{array}$ \\
\midrule

GFM &
$\begin{array}{c}
\tfrac{1}{m_\mathrm{p}\omega_\mathrm{c}}>0,\;\;
\tfrac{1}{m_\mathrm{p}}>0
\end{array}$ &
$\begin{array}{c}
\tfrac{1}{m_\mathrm{p}\omega_\mathrm{c}}<0,\;\;
\tfrac{1}{m_\mathrm{p}}<0
\end{array}$ \\
\midrule

GFL &
$\begin{array}{c}
\tfrac{2\zeta_\mathrm{PLL}}{m_\mathrm{p}\omega_{\rm PLL}}
+\tfrac{1}{\omega_{\rm PLL}^2}\psi > 0 \\[4pt]
\tfrac{1}{m_\mathrm{p}}
+\tfrac{2\zeta_\mathrm{PLL}}{\omega_{\rm PLL}}\psi > 0
\end{array}$ &
$\begin{array}{c}
\tfrac{2\zeta_\mathrm{PLL}}{m_\mathrm{p}\omega_{\rm PLL}}
-\tfrac{1}{\omega_{\rm PLL}^2}\psi < 0 \\[4pt]
\tfrac{1}{m_\mathrm{p}}
-\tfrac{2\zeta_\mathrm{PLL}}{\omega_{\rm PLL}}\psi < 0
\end{array}$ \\
\midrule
\bottomrule
\end{tabular}
\end{table}

\begin{figure*}
    \centering
    \includegraphics[scale= 1.0]{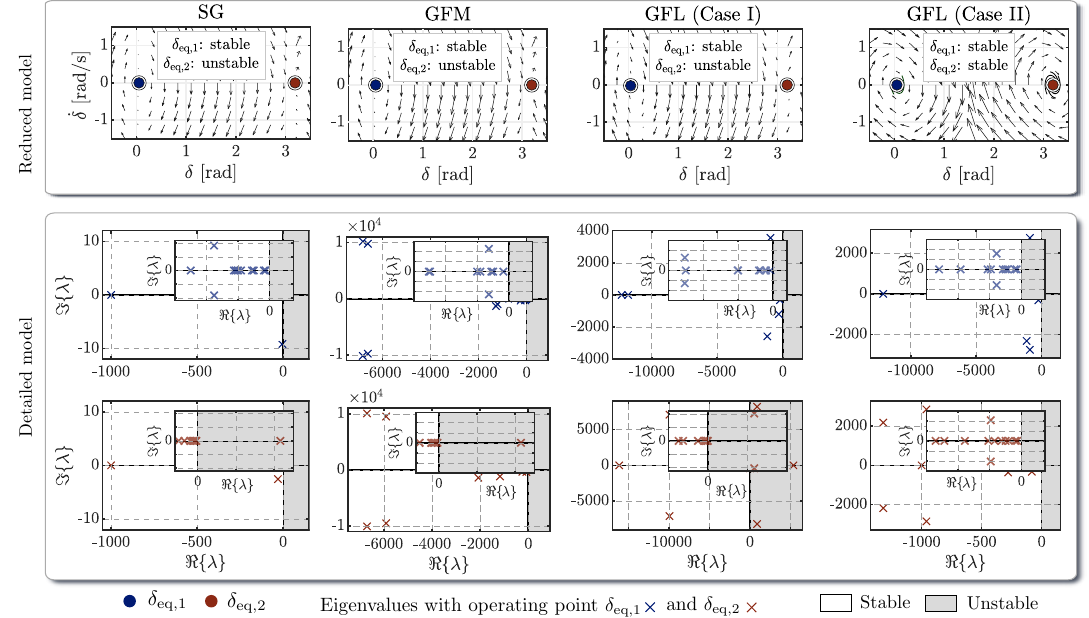}
    \caption{\emph{Top} (Reduced model) Phase portraits in the $(\delta,\,\dot{\delta})$ plane, marking $\delta_{\mathrm{eq},1}$ and $\delta_{\mathrm{eq},2}$. \emph{Bottom} (Detailed model) Eigenvalues about $\delta_{\mathrm{eq},1}$ and $\delta_{\mathrm{eq},2}$; the unshaded (shaded) region is the stable (unstable) half-plane. The SG and GFM are stable at $\delta_{\mathrm{eq},1}$ and unstable at $\delta_{\mathrm{eq},2}$; the GFL follows this in Case~I. In Case~II, different values for $\omega_{\mathrm{PLL}}$ and  $\zeta_\mathrm{PLL}$, stabilizes $\delta_{\mathrm{eq},2}$ (validated by the absence of right-half-plane eigenvalues).}
    \label{fig:stability}
\end{figure*}

\subsection{Examining Generators \& Inverters} \label{eq:resources}
In what follows, we examine reduced-order models for SGs, GFL inverters, and GFM inverters, and recognize that they align with~\eqref{eq:angle-power-tf}. Consequently, in a single-unit infinite-bus setting, each admits the equilibria reported in~\eqref{eq:delta_eq12}. Parametric conditions for the two equilibria are summarized in Table~\ref{tab:stability}. (Table~\ref{tab:stability_full} in Appendix~A provides helpful intermediate details.)
\subsubsection{Synchronous Generator}
Adopting the classical model (acknowledging only simplified rotor dynamics): 
\begin{equation} \label{eq:H-SG}
    \mathcal{G}_\mathrm{SG}(s) = \tfrac{2H}{\omega_\mathrm{s}}s^2 + \tfrac{D}{\omega_\mathrm{s}}s,
\end{equation}
where $H$ is inertia constant and $D$ is the damping constant~\cite{Ajala-2020}. From the conditions in Table~\ref{tab:stability}; for $H,D>0$, $\delta_{\mathrm{eq},1}$ is always stable, while $\delta_{\mathrm{eq},2}$ is always unstable. 

\subsubsection{Grid-following IBR}
The GFL IBR model includes a synchronous reference frame (SRF) phase-locked loop (PLL) implemented with a Proportional Integral (PI) controller with gains $k_\mathrm{p, PLL}$ and $k_\mathrm{i, PLL}$. Suppose the $\mathrm{q}$-axis component of the terminal voltage is regulated to zero; the estimated frequency in quasi-steady state, $\Omega_\mathrm{GFL}(s) = \mathscr{L}\{\omega_\mathrm{GFL}(t)\}$ is:~\cite {Green-GFL_GFM-2023}  
\begin{equation} \label{eq:GFL-omega-relationship}
   \Omega_\mathrm{GFL}(s) = \frac{2 \zeta_\mathrm{PLL} \omega_\mathrm{PLL}s + \omega^2_\mathrm{PLL}}{s^2 + 2\zeta_\mathrm{PLL}\omega_\mathrm{PLL}s + \omega_\mathrm{PLL}} \Omega(s),
\end{equation}
where $\zeta_\mathrm{PLL}$ and $\omega_\mathrm{PLL}$ are the damping coefficient and bandwidth of the PLL related to the PI gains, $k_\mathrm{p, PLL},k_\mathrm{i, PLL}$, and terminal voltage, $E$, via: $k_\mathrm{i, PLL} E = \omega_{\mathrm{PLL}}^2$ and $k_\mathrm{p, PLL} E = 2\zeta_\mathrm{PLL} \omega_{\mathrm{PLL}}$. The GFL IBR is set up to provide fast frequency response, with $m_\mathrm{p}$ denoting the active power versus frequency droop gain. The active-power balance equation then reads: $\mathcal{P}_\delta(s) = \mathcal{P}^\star(s) - \tfrac{1}{m_\mathrm{p}} \big(\Omega_\mathrm{GFL} (s) - \tfrac{1}{s}\omega_\mathrm{s}\big)$. Substituting $\Omega_\mathrm{GFL}(s)$ from~\eqref{eq:GFL-omega-relationship}, yields a model of the form~\eqref{eq:delta-p-Laplace}:~\cite {Green-GFL_GFM-2023} 
\begin{equation} \label{eq:H-GFL}
    \mathcal{G}_\mathrm{GFL} (s) =  \frac{\frac{2\zeta_\mathrm{PLL}}{m_\mathrm{p}\omega_{\mathrm{PLL}}} s^2 + \tfrac{1}{m_\mathrm{p}}s}{\frac{1}{\omega_{\mathrm{PLL}}^2}s^2 + \frac{2\zeta_\mathrm{PLL}}{\omega_{\mathrm{PLL}}} s + 1}.
\end{equation}
Consider the conditions reported in Table~\ref{tab:stability}. For positive $\omega_{\mathrm{PLL}}$, $\zeta_\mathrm{PLL}$, $m_\mathrm p$, and recognizing from~\eqref{eq:SEP_UEP_partial} that $\psi >0$, we see that $\delta_{\mathrm{eq},1}$ is guaranteed to be small-signal stable. Rearranging terms, the stability condition for $\delta_{\mathrm{eq},2}$ can be expressed as $\tfrac{2\zeta_\mathrm{PLL}\,\omega_{\mathrm{PLL}}}{m_\mathrm p}<\psi$ and $\tfrac{\omega_{\mathrm{PLL}}}{2\zeta_\mathrm{PLL}\,m_\mathrm p}<\psi$. While these can be satisfied by judiciously tuned parameters, the corresponding dynamics of the PLL are rendered unreasonably sluggish. 
 
\subsubsection{Grid-forming IBR}
Modeling the active-power versus frequency droop behavior with a low-pass filter on the power measurement yields the transfer function: 
\begin{equation} \label{eq:H-GFM}
    \mathcal{G}_\mathrm{GFM}(s) = \tfrac{1}{m_\mathrm{p}\omega_\mathrm{c}} s^2 + \tfrac{1}{m_\mathrm{p}} s,
\end{equation}
where $m_\mathrm{p}$ is the active power versus frequency droop gain, and $\omega_\mathrm{c}$ is the cut-off frequency for the first-order low-pass filter on the active-power measurement at the terminals. With the equivalence reported in~\cite{Suul-Equivalence-2014}, the above model also applies to virtual synchronous generator-based primary controls. On the other hand, virtual oscillator-based primary controls---in their original incarnations---do not align due to the conspicuous absence of inertia~\cite{Ajala-2023}. (Follow-up efforts appear to have included inertia contributions; see, e.g.,~\cite{VOC-Review-2022}.) From the conditions reported in Table~\ref{tab:stability}; we see that for $m_\mathrm{p},\omega_\mathrm{c}>0$, $\delta_{\mathrm{eq},1}$ is always stable, while $\delta_{\mathrm{eq},2}$ is always unstable. 

\section{Validation} \label{sec:Validation}
Simulations are performed in the open-source toolbox \texttt{PowerSimulationsDynamics.jl}~\cite{lara2024powersimulationsdynamicsjlopensource} with models available in~\cite{damola_ajeyemi_2026_20779024}. We test the analysis on a single-unit infinite-bus system with line impedance $0.005+\jmath0.10~\mathrm{[pu]}$, $E=V_\infty=1.0~\mathrm{[pu]}$, and $P^\star=0.5~\mathrm{[pu]}$. Note that the interconnection is predominantly inductive with $X/R=20$. Furthermore, all validation cases use the same single-unit infinite-bus network so that differences in small-signal behavior can be attributed to the unit dynamics rather than to changes in the operating point or network strength. For these parameters,~\eqref{eq:delta_eq12} gives two equilibria: $\delta_{\mathrm{eq},1}=2.87^\circ$ and $\delta_{\mathrm{eq},2}=182.86^\circ$. The following parameters are adopted for the SG's reduced-order model: $H=4~\mathrm{[s]}$, $D=40$; and for the GFM: $m_\mathrm{p}=0.025$, $\omega_\mathrm{c}=5~\mathrm{[rad\cdot\mathrm{s}^{-1}]}$. These values are picked to ensure alignment of the $a_i,b_i$ coefficients in the angle-to-power transfer functions. Two cases are considered for the GFL. In Case~I, $\zeta_\mathrm{PLL}=2.711$, $\omega_\mathrm{PLL}=27.11~\mathrm{[rad \cdot s^{-1}]}$, $m_\mathrm{p}=0.025$; which impose nominal tracking speed. In Case~II, $\zeta_\mathrm{PLL}=0.21$, $\omega_\mathrm{PLL}=2.164~\mathrm{[rad \cdot s^{-1}]}$, $m_\mathrm{p}=0.025$, which renders the PLL sluggish but will be useful to demonstrate the stabilizability of $\delta_{\mathrm{eq},2}$. (See Table~\ref{tab:reduced_order_parameters_all_resources} in Appendix~B for detailed model parameters.)

Figure~\ref{fig:stability} \textit{Top} (Reduced model) shows phase portraits from the second-order models. Figure~\ref{fig:stability} \textit{Bottom} (Detailed model) shows eigenvalues of the detailed $19\mathrm{th}$ order SG, $12\mathrm{th}$ order GFL, and $14\mathrm{th}$ order GFM models (introductory remarks document the dynamics considered) linearized around $\delta_{\mathrm{eq},1}$ and $\delta_{\mathrm{eq},2}$. (Parameters are listed in Appendix~C, and for the chosen parameters, equilibria across the reduced and detailed models are about the same.) We make two broad observations:
\begin{itemize}[leftmargin=*]
    \item For the SG and GFM, the phase portraits for the reduced model show a stable node at $\delta_{\mathrm{eq},1}$ and an unstable node at $\delta_{\mathrm{eq},2}$. Eigenvalues confirm $\delta_{\mathrm{eq},1}$ is stable while $\delta_{\mathrm{eq},2}$ is unstable.
    \item For the GFL, parameters in Case~I induce the same behavior as the SG and GFM reported above. On the other hand, the parameters picked for Case~II render $\delta_{\mathrm{eq},2}$ stable. 
\end{itemize}

\section{Concluding Remarks \& Future Work} \label{sec:ConclusionsFutureWork}
We proposed a unified swing-equation model for SGs, GFL inverters with fast frequency response, and droop- and virtual synchronous generator-based GFM inverters connected to an infinite bus. Parametric conditions for small-signal stability of angle equilibria were derived using reduced-order representations for the angle-to-power transfer functions. Eigenvalues obtained by linearizing higher-order dynamics were provided to justify modeling assumptions and validate inferences from the lower-order representations. As part of future work, we aim to scale the analysis to network-wide stability conditions and bring in other resource models into the modeling framework. 

\appendix
\subsection{Unified Swing-equation Model Parameters and Small-signal Stability Conditions}
Table~\ref{tab:stability_full} maps each unit's reduced-order angle-to-power transfer function $\mathcal{G}(s)$ onto the unified form in~\eqref{eq:angle-power-tf}. For the SG, GFM, and GFL, it lists the numerator and denominator coefficients $b_2$, $b_1$, $a_2$, $a_1$, $a_0$. The SG and GFM have $a_2 = a_1 = 0$ and $a_0 = 1$, so their inertia and damping constants $M(\delta_{\mathrm{eq}})$ and $D(\delta_{\mathrm{eq}})$ depend only on the unit. The GFL has a nontrivial denominator, and the power-flow sensitivity $\frac{\partial}{\partial\delta}P(\delta_{\mathrm{eq}})$ enters both $M(\delta_{\mathrm{eq}})$ and $D(\delta_{\mathrm{eq}})$ through it. The last two columns apply Theorem~1. They give the sign conditions of~\eqref{eq:SEPstability}--\eqref{eq:UEPstability} that govern small-signal stability at the low-angle equilibrium $\delta_{\mathrm{eq},1}$ and the high-angle equilibrium $\delta_{\mathrm{eq},2}$. 

\begin{table*}[t]
\centering
\caption{Unified swing-equation model parameters for SG, GFL, and GFM units; and small-signal stability conditions for the two equilibria.}
\label{tab:stability_full}
\renewcommand{\arraystretch}{1.0}{
\begin{tabular}{c||cccccc|ccc}
\toprule
\midrule
\textbf{unit} & \boldmath{$b_2$} &  \boldmath{$b_1$} &  \boldmath{$b_0$} &  \boldmath{$a_2$} & \boldmath{ $a_1$} &   \boldmath{$a_0$} &
 \boldmath{$M(\delta_{\rm eq})$} &  \boldmath{$D(\delta_{\rm eq})$} &  \boldmath{$K(\delta_{\rm eq})$} \\
\midrule
SG &
$\tfrac{2H}{\omega_\mathrm{s}}$ & $\tfrac{D}{\omega_\mathrm{s}}$ & $0$ & $0$ & $0$ & $1$ &
$\tfrac{2H}{\omega_\mathrm{s}}$ & $\tfrac{D}{\omega_\mathrm{s}}$ &
$\tfrac{\partial}{\partial \delta}P(\delta_{\rm eq})$ \\
\midrule
GFM &
$\tfrac{1}{m_\mathrm{p}\omega_\mathrm{c}}$ & $\tfrac{1}{m_\mathrm{p}}$ & $0$ & $0$ & $0$ & $1$ &
$\tfrac{1}{m_\mathrm{p}\omega_\mathrm{c}}$ & $\tfrac{1}{m_\mathrm{p}}$ &
$\tfrac{\partial}{\partial \delta}P(\delta_{\rm eq})$ \\
\midrule
GFL &
$\tfrac{2\zeta_\mathrm{PLL}}{m_\mathrm{p}\omega_{\rm PLL}}$ & $\tfrac{1}{m_\mathrm{p}}$ & $0$ &
$\tfrac{1}{\omega_{\rm PLL}^2}$ & $\tfrac{2\zeta_\mathrm{PLL}}{\omega_{\rm PLL}}$ & $1$ &
$\begin{array}{c} \tfrac{2\zeta_\mathrm{PLL}}{m_\mathrm{p}\omega_{\rm PLL}}+\tfrac{1}{\omega_{\rm PLL}^2}\tfrac{\partial}{\partial \delta}P(\delta_{\rm eq}) \end{array}$ &
$\begin{array}{c} \tfrac{1}{m_\mathrm{p}}+\tfrac{2\zeta_\mathrm{PLL}}{\omega_{\rm PLL}}\tfrac{\partial}{\partial \delta}P(\delta_{\rm eq}) \end{array}$ &
$\tfrac{\partial}{\partial \delta}P(\delta_{\rm eq})$ \\
\midrule
\bottomrule
\end{tabular}
}
\end{table*}

\subsection{Reduced-order Model Parameters and Coefficients of Unified Swing-equation Model}
Table~\ref{tab:reduced_order_parameters_all_resources} reports the resource-specific data in three blocks. The first block lists the reduced-order unit and controller parameters for each case: the SG inertia and damping, the GFM droop gain and filter cutoff, and the PLL bandwidth, damping, and PI gains for the two GFL cases. The second block gives the resulting transfer-function coefficients $b_i$ and $a_i$, computed from these parameters using the expressions in Table~\ref{tab:stability_full}. The parameters were chosen so that the $b_i$ coincide across the SG, GFM, and GFL. This isolates the effect of the denominator terms $a_2$ and $a_1$, which are unique to the GFL. The third block evaluates the swing coefficients $M(\delta_{\mathrm{eq}})$, $D(\delta_{\mathrm{eq}})$, and $K(\delta_{\mathrm{eq}})$ at both equilibria using~\eqref{eq:MDK} and the sensitivity sign in~\eqref{eq:SEP_UEP_partial}. For the SG, GFM, and GFL Case~I, all three coefficients are positive at $\delta_{\mathrm{eq},1}$ (stable), while $K(\delta_{\mathrm{eq},2}) < 0$ (unstable). For GFL Case~II, the sluggish-PLL parameters drive $M(\delta_{\mathrm{eq}})$ and $D(\delta_{\mathrm{eq}})$ negative at $\delta_{\mathrm{eq},2}$ along with $K$. This restores the sign agreement of~\eqref{eq:stability} and stabilizes the high-angle equilibrium, the distinctive behavior emphasized in the paper.

\begin{table}[t!]
\centering
\caption{Reduced-order model parameters and coefficients for unified swing-equation model for SG, GFM, and GFL units}
\label{tab:reduced_order_parameters_all_resources}
\renewcommand{\arraystretch}{1.5}
\scriptsize
\setlength{\tabcolsep}{3pt}
\begin{tabular}{@{}c c c c c c@{}}
\hline\hline
\textbf{Quantity}
& \textbf{Unit}
& \textbf{SG}
& \textbf{GFM}
& \textbf{\shortstack{\rule{0pt}{2ex}GFL\\Case I }}
& \textbf{\shortstack{GFL\\Case II}} \\
\hline
\multicolumn{6}{c}{\textbf{I. Unit and controller parameters}} \\
\hline
$\omega_{\mathrm{s}}$
& $\mathrm{rad/s}$ & $2\pi60$ & $2\pi60$ & $2\pi60$ & $2\pi60$ \\
$H$
& $\mathrm{s}$ & $4.0$ & -- & -- & -- \\
$D$
& -- & $40.0$ & -- & -- & -- \\
$m_{\mathrm{p}}$
& -- & -- & $0.025$ & $0.025$ & $0.025$ \\
$\omega_{\mathrm{c}}$
& $\mathrm{rad/s}$ & -- & $5.0$ & -- & -- \\
$\omega_{\mathrm{PLL}}$
& $\mathrm{rad/s}$ & -- & -- & $27.118$ & $2.165$ \\
$\zeta_{\mathrm{PLL}}$
& -- & -- & -- & $2.711$ & $0.216$ \\
$k_{\mathrm{p, PLL}}$
& $\mathrm{rad/s/pu}$ & -- & -- & $147.078$ & $0.938$ \\
$k_{\mathrm{i, PLL}}$
& $\mathrm{{rad/s^2/pu}}$ & -- & -- & $735.392$ & $4.690$ \\
\hline\hline
\multicolumn{6}{c}{\textbf{II. Angle-to-power transfer-function coefficients}} \\
\hline
$b_2 $
& $\mathrm{s^2}$
& $\tfrac{2H}{\omega_{\mathrm{s}}}$
& $\tfrac{1}{m_{\mathrm{p}}\omega_{\mathrm{c}}}$
& $\tfrac{2\zeta_{\mathrm{PLL}}}{m_{\mathrm{p}}\omega_{\mathrm{PLL}}}$
& $\tfrac{2\zeta_{\mathrm{PLL}}}{m_{\mathrm{p}}\omega_{\mathrm{PLL}}}$ \\
&
& $=0.021$ & $=0.021$ & $=0.021$ & $=0.021$ \\
$b_1$
& $\mathrm{s}$
& $\tfrac{D}{\omega_{\mathrm{s}}}$
& $\tfrac{1}{m_{\mathrm{p}}}$
& $\tfrac{1}{m_{\mathrm{p}}}$
& $\tfrac{1}{m_{\mathrm{p}}}$ \\
&
& $=0.106$ & $=0.106$ & $=0.106$ & $=0.106$ \\
\hline
$a_2$
& $\mathrm{s^2}$
& $0$
& $0$
& $\tfrac{1}{\omega_{\mathrm{PLL}}^2}$
& $\tfrac{1}{\omega_{\mathrm{PLL}}^2}$ \\
&
&
&
& $=0.0013$
& $=0.21$ \\
$a_1$
& $\mathrm{s}$
& $0$
& $0$
& $\tfrac{2\zeta_{\mathrm{PLL}}}{\omega_{\mathrm{PLL}}}$
& $\tfrac{2\zeta_{\mathrm{PLL}}}{\omega_{\mathrm{PLL}}}$ \\
&
&
&
& $=0.2$
& $=0.2$ \\
$a_0$
& -- & $1$ & $1$ & $1$ & $1$ \\
\hline\hline
\multicolumn{6}{c}{\textbf{III. Equivalent swing coefficients at the two equilibria}} \\
\hline
$M(\delta_{\mathrm{eq},1})$
& $\mathrm{s^2}$ & $0.021$ & $0.021$ & $0.034$ & $2.150$ \\
$D(\delta_{\mathrm{eq},1})$
& $\mathrm{s}$ & $0.106$ & $0.106$ & $2.103$ & $2.103$ \\
$K(\delta_{\mathrm{eq},1})$
& -- & $9.987$ & $9.987$ & $9.987$ & $9.987$ \\
\hline
$M(\delta_{\mathrm{eq},2})$
& $\mathrm{s^2}$ & $0.021$ & $0.021$ & $0.007$ & $-2.108$ \\
$D(\delta_{\mathrm{eq},2})$
& $\mathrm{s}$ & $0.106$ & $0.106$ & $-1.891$ & $-1.891$ \\
$K(\delta_{\mathrm{eq},2})$
& -- & $-9.987$ & $-9.987$ & $-9.987$ & $-9.987$ \\
\hline\hline
\end{tabular}
\end{table}

\subsection{Detailed-model Parameters}
\label{sec:detailed_parameters}
Tables~\ref{tab:sg_detailed}--\ref{tab:gfl_detailed} list the parameters of the
higher-order unit models linearized to produce the eigenvalue plots in
Fig.~\ref{fig:stability}: the $19\mathrm{th}$-order SG, the $14\mathrm{th}$-order GFM, and the
$12\mathrm{th}$-order GFL. All quantities are in per unit on the $100\, \mathrm{MVA}$, $60\, \mathrm{Hz}$ base frequency.
For the GFL, only the PLL gains differ between Case~I (fast PLL) and Case~II (sluggish PLL); every other entry is shared. Detailed models and associated parameters can be found in~\cite{lara2024powersimulationsdynamicsjlopensource}.

\begin{table}[t!]
\centering
\caption{Detailed SG model parameters} 
\label{tab:sg_detailed}
\renewcommand{\arraystretch}{1.2}
\scriptsize
\setlength{\tabcolsep}{3pt}
\begin{tabular}{c p{4.5cm} c}
\hline\hline
\textbf{Symbol} & \textbf{Description} & \textbf{Value} \\
\hline
\multicolumn{3}{c}{\textbf{Round-rotor machine}} \\
\hline
$R_\mathrm{s}$   & Stator (armature) resistance               & $0.020~\mathrm{pu}$ \\
$X_\ell$         & Stator leakage reactance                   & $0.200~\mathrm{pu}$ \\
$X_\mathrm{d}$   & $d$-axis synchronous reactance             & $1.800~\mathrm{pu}$ \\
$X_\mathrm{q}$   & $q$-axis synchronous reactance             & $1.700~\mathrm{pu}$ \\
$X_\mathrm{d}'$  & $d$-axis transient reactance               & $0.300~\mathrm{pu}$ \\
$X_\mathrm{q}'$  & $q$-axis transient reactance               & $0.550~\mathrm{pu}$ \\
$X_\mathrm{d}''$ & $d$-axis subtransient reactance ($=X_\mathrm{q}''$) & $0.250~\mathrm{pu}$ \\
$T_\mathrm{d0}'$ & $d$-axis transient open-circuit time const.  & $6.00~\mathrm{s}$ \\
$T_\mathrm{q0}'$ & $q$-axis transient open-circuit time const.  & $0.80~\mathrm{s}$ \\
$T_\mathrm{d0}''$& $d$-axis subtransient open-circuit time const. & $0.050~\mathrm{s}$ \\
$T_\mathrm{q0}''$& $q$-axis subtransient open-circuit time const. & $0.035~\mathrm{s}$ \\
\hline
\multicolumn{3}{c}{\textbf{Single-mass shaft}} \\
\hline
$H$ & Inertia constant            & $4.0~\mathrm{s}$ \\
$D$ & Mechanical damping constant & $40.0$ \\
\hline
\multicolumn{3}{c}{\textbf{Type-I AVR}} \\
\hline
$K_\mathrm{a}$ & Regulator (amplifier) gain        & $1.0$ \\
$K_\mathrm{e}$ & Exciter self-excitation constant  & $1.0$ \\
$K_\mathrm{f}$ & Rate feedback gain                & $0.0$ \\
$T_\mathrm{a}$ & Amplifier time constant           & $0.015~\mathrm{s}$ \\
$T_\mathrm{e}$ & Exciter time constant             & $0.50~\mathrm{s}$ \\
$T_\mathrm{f}$ & Rate feedback time constant       & $1.0~\mathrm{s}$ \\
$T_\mathrm{r}$ & Voltage transducer time constant  & $0.010~\mathrm{s}$ \\
$V_\mathrm{a}^{\lim}$ & Regulator output limits    & $(-10,\,10)~\mathrm{pu}$ \\
$A_\mathrm{e},B_\mathrm{e}$ & Exciter saturation constants & $(0.001,\,1.0)$ \\
\hline
\multicolumn{3}{c}{\textbf{Type-I turbine-governor}} \\
\hline
$R$ & Governor speed droop                 & $0.045~\mathrm{pu}$ \\
$T_\mathrm{s}$ & Speed-relay time constant   & $0.50~\mathrm{s}$ \\
$T_\mathrm{c}$ & Servomotor time constant    & $0.001~\mathrm{s}$ \\
$T_3,T_4,T_5$ & Transient/reheat time const. & $0.001~\mathrm{s}$ \\
$v^{\lim}$ & Valve position limits           & $(-5,\,5)~\mathrm{pu}$ \\
\hline
\multicolumn{3}{c}{\textbf{IEEEST power system stabilizer}} \\
\hline
$K_\mathrm{s}$ & Stabilizer gain               & $10.0$ \\
$T_1,T_3$ & Lead time constants               & $0.25~\mathrm{s}$ \\
$T_2,T_4$ & Lag time constants                & $0.02~\mathrm{s}$ \\
$T_5$ & Washout (reset) time constant         & $1.0~\mathrm{s}$ \\
$T_6$ & Input filter time constant            & $0.08~\mathrm{s}$ \\
$L_\mathrm{s}^{\lim}$ & Stabilizer output limits & $(-0.05,\,0.05)~\mathrm{pu}$ \\
$V_\mathrm{cu},V_\mathrm{cl}$ & Terminal-voltage cutoffs & $(\pm1.25)~\mathrm{pu}$ \\
\hline\hline
\end{tabular}
\end{table}

\begin{table}[t!]
\centering
\caption{Detailed GFM model parameters}
\label{tab:gfm_detailed}
\renewcommand{\arraystretch}{1.2}
\scriptsize
\setlength{\tabcolsep}{3pt}
\begin{tabular}{c p{4.5cm} c}
\hline\hline
\textbf{Symbol} & \textbf{Description} & \textbf{Value} \\
\hline
\multicolumn{3}{c}{\textbf{Outer control: virtual inertia (active power)}} \\
\hline
$T_\mathrm{a}$     & Virtual inertia constant ($\equiv 2H_\mathrm{vsm}$) & $8.0~\mathrm{s}$ \\
$k_\mathrm{d}$     & Frequency-damping gain   & $19.623$ \\
$k_\mathrm{\omega}$& Frequency-droop gain     & $20.377$ \\
\hline
\multicolumn{3}{c}{\textbf{Outer control: reactive-power droop}} \\
\hline
$k_\mathrm{q}$  & Reactive ($Q$-$V$) droop gain        & $0.20~\mathrm{pu}$ \\
$\omega_\mathrm{f}$ & $Q$ measurement filter cutoff   & $1000~\mathrm{rad/s}$ \\
\hline
\multicolumn{3}{c}{\textbf{Inner voltage/current control}} \\
\hline
$k_\mathrm{pv}$  & Voltage-loop proportional gain & $1.0$ \\
$k_\mathrm{iv}$  & Voltage-loop integral gain     & $1000$ \\
$k_\mathrm{pc}$  & Current-loop proportional gain & $3.0$ \\
$k_\mathrm{ic}$  & Current-loop integral gain     & $600$ \\
$\omega_\mathrm{ad}$ & Active-damping filter cutoff & $2000~\mathrm{rad/s}$ \\
$k_\mathrm{ad}$  & Active-damping gain            & $2.0$ \\
\hline
\multicolumn{3}{c}{\textbf{$LCL$ filter}} \\
\hline
$l_\mathrm{f}$ & Converter-side inductance & $0.080~\mathrm{pu}$ \\
$r_\mathrm{f}$ & Converter-side resistance & $0.003~\mathrm{pu}$ \\
$c_\mathrm{f}$ & Shunt capacitance         & $0.074~\mathrm{pu}$ \\
$l_\mathrm{g}$ & Grid-side inductance      & $0.200~\mathrm{pu}$ \\
$r_\mathrm{g}$ & Grid-side resistance      & $0.010~\mathrm{pu}$ \\
\hline\hline
\end{tabular}
\end{table}

\begin{table}[h!]
\centering
\caption{Detailed GFL model parameters}
\label{tab:gfl_detailed}
\renewcommand{\arraystretch}{1.2}
\scriptsize
\setlength{\tabcolsep}{3pt}
\begin{tabular}{c p{3.0cm} c c}
\hline\hline
\textbf{Symbol} & \textbf{Description} & \textbf{Case I} & \textbf{Case II} \\
\hline
\multicolumn{4}{c}{\textbf{Outer control: active-power PI}} \\
\hline
$K_\mathrm{p, p}$ & Active-power prop. gain & \multicolumn{2}{c}{$2.0$} \\
$K_\mathrm{i, p}$ & Active-power int. gain  & \multicolumn{2}{c}{$30$} \\
$\omega_\mathrm{z, p}$ & $P$ measurement cutoff & \multicolumn{2}{c}{$500~\mathrm{rad/s}$} \\
\hline
\multicolumn{4}{c}{\textbf{Outer control: reactive-power PI}} \\
\hline
$K_\mathrm{p, q}$ & Reactive-power prop. gain & \multicolumn{2}{c}{$2.0$} \\
$K_\mathrm{i, q}$ & Reactive-power int. gain  & \multicolumn{2}{c}{$30$} \\
$\omega_\mathrm{f, q}$ & $Q$ measurement cutoff & \multicolumn{2}{c}{$500~\mathrm{rad/s}$} \\
\hline
\multicolumn{4}{c}{\textbf{Inner current control}} \\
\hline
$k_\mathrm{pc}$  & Current-loop prop. gain & \multicolumn{2}{c}{$3.0$} \\
$k_\mathrm{ic}$  & Current-loop int. gain  & \multicolumn{2}{c}{$600$} \\
$k_\mathrm{ffv}$ & Voltage feed-forward gain & \multicolumn{2}{c}{$0.0$} \\
\hline
\multicolumn{4}{c}{\textbf{PLL and FFR gains}} \\
\hline
$m_\mathrm{p}$  & Droop gain & \multicolumn{2}{c}{$0.025$}\\
$\omega_\mathrm{PLL}$ & PLL bandwidth & \multicolumn{2}{c}{$27.11~\mathrm{rad/s}$} \\
$k_\mathrm{p, PLL}$ & PLL proportional gain & $147.078$ & $0.938$ \\
$k_\mathrm{i, PLL}$ & PLL integral gain     & $735.392$ & $4.690$ \\
\hline
\multicolumn{4}{c}{\textbf{$LCL$ filter}} \\
\hline
$l_\mathrm{f}$ & Converter-side inductance & \multicolumn{2}{c}{$0.080~\mathrm{pu}$} \\
$r_\mathrm{f}$ & Converter-side resistance & \multicolumn{2}{c}{$0.003~\mathrm{pu}$} \\
$c_\mathrm{f}$ & Shunt capacitance         & \multicolumn{2}{c}{$0.074~\mathrm{pu}$} \\
$l_\mathrm{g}$ & Grid-side inductance      & \multicolumn{2}{c}{$0.200~\mathrm{pu}$} \\
$r_\mathrm{g}$ & Grid-side resistance      & \multicolumn{2}{c}{$0.010~\mathrm{pu}$} \\
\hline\hline
\end{tabular}
\end{table}

\bibliographystyle{IEEEtran}

\end{document}